\documentclass{article}

\usepackage{microtype}
\usepackage{graphicx}
\usepackage{subfigure}
\usepackage{booktabs} 
\usepackage{multirow}
\usepackage{bm}
\usepackage{float}
\usepackage{comment}

 \usepackage[nohyperref,accepted]{icml2023}

\usepackage{amsmath}
\usepackage{amssymb}
\usepackage{mathtools}
\usepackage{amsthm}

\usepackage[capitalize,noabbrev]{cleveref}

\theoremstyle{plain}
\newtheorem{theorem}{Theorem}[section]

\theoremstyle{definition}
\newtheorem{definition}[theorem]{Definition}

\theoremstyle{remark}

\usepackage[textsize=tiny]{todonotes}

\begin{document}

\twocolumn[
\icmltitle{Online Local Differential Private Quantile Inference via Self-normalization}



\icmlsetsymbol{equal}{*}

\begin{icmlauthorlist}
\icmlauthor{Yi Liu}{yyy}
\icmlauthor{Qirui Hu}{comp}
\icmlauthor{Lei Ding}{yyy}
\icmlauthor{Bei Jiang}{yyy}
\icmlauthor{Linglong Kong}{yyy}

\end{icmlauthorlist}

\icmlaffiliation{yyy}{Department of Mathematical and Statistical Sciences, University of Alberta, Edmonton, Canada}
\icmlaffiliation{comp}{Center for Statistical Science, Department of Industrial Engineering, Tsinghua University, Beijing, China}

\icmlcorrespondingauthor{Linglong Kong}{lkong@ualberta.ca}

\icmlkeywords{Machine Learning, ICML}

\vskip 0.3in
]



\printAffiliationsAndNotice{}  

\begin{abstract}
Based on binary inquiries, we developed an algorithm to estimate population quantiles under Local Differential Privacy (LDP). By self-normalizing, our algorithm provides asymptotically normal estimation with valid inference, resulting in tight confidence intervals without the need for nuisance parameters to be estimated. Our proposed method can be conducted fully online, leading to high computational efficiency and minimal storage requirements with $\mathcal{O}(1)$ space. We also proved an optimality result by an elegant application of one central limit theorem of Gaussian Differential Privacy (GDP) when targeting the frequently encountered median estimation problem. With mathematical proof and extensive numerical testing, we demonstrate the validity of our algorithm both theoretically and experimentally.

\end{abstract}

\section{Introduction}

Personal data is currently widely used for various purposes, such as facial recognition, personalized advertising, medical trials, and recommendation systems to name a few.  While there are potential benefits, it is important also to consider the risks associated with handling sensitive personal information.  For instance, research on diabetes can provide valuable insights that may benefit society as a whole in the long term. However, it is crucial to keep in mind that participants may suffer direct consequences if their data is not properly protected through controlled disclosure, such as a rise in health insurance premiums.

The concept of Differential Privacy (DP; \citealp{dwork2006our}) has been successful in providing a rigorous condition for controlled disclosure by bounding the change in the distribution of outputs of a query made on a dataset under the alteration of one data point. This has led to a vast amount of literature under the umbrella of DP, resulting in various generalizations, tools, and applications. However, while enjoying the mathematically solid guarantee of DP and its variants, concerns about a weak link in the process, the trusted curator, are beginning to arise.

The use of trusted curators undermines the spirit of the solid cryptographic level of privacy protection that DP provides. This risk is not limited to information security breaches and rogue researchers but also includes legal proceedings where researchers may be compelled to hand over the raw data, breaking the initial promise made to DP at the time of data collection. Two concepts, Local Differential Privacy (LDP) and pan-DPs, are proposed as solutions. The pan-DP directly counters this issue by solidifying the algorithm to withstand multiple announced intrusions (subpoenas) or one unannounced intrusion (hackers). The concept of LDP was first introduced formally by \citep{kasiviswanathan2011can}, but its early form can be traced back to \citep{evfimievski2003limiting} and \citep{warner1965randomized}, in the name of "amplification" and "randomized response survey," respectively. 

In LDP settings, the sensitive information never leaves the control of the users unprotected. The users encode and alter their data locally before sending them to an untrusted central data server for further analysis and computation. Recently, in \citep{pmlr-v125-amin20a} unveiled a connection between pan-DP and LDP by considering variants of pan-DP framework that can defend against multiple unannounced intrusions. Surprisingly, this requirement can only be fulfilled if the data is scrambled before it leaves the owner's control, which goes back to the definition of LDP.  For better privacy protection, many big tech companies have already implemented LDP into their products, such as Google \cite{erlingsson2014rappor} and Microsoft \cite{ding2017collecting}.

This discovery rekindled the research interest in LDP. Researchers have begun to consider fundamental statistical problems, such as estimating parameters, modeling, and hypothesis testing under this constraint. The quantiles, including the median, are basic summary statistics that have been widely studied within the framework of differential privacy. Early research in this area includes the estimation of quantiles under the central DP setting, as presented in \citep{dwork2009differential} and \citep{lei2011differentially}. More recent advancements, such as \citep{smith2011privacy}, have proposed a rate-optimal sample quantile estimator that does not rely on the evaluation of histograms. \citep{pmlr-v139-gillenwater21a} further extended this research by estimating multiple quantiles simultaneously. Despite these advances, the quantile estimation under the central DP setting remains an active area of research, with new work in various applications such as \citep{alabi2022bounded} and \citep{ben2022archimedes}.

In the central DP setting, a trusted curator can acquire the actual sample quantiles and other summary statistics, with the only limitation being that the release of the output must conform to the DP condition. However, under the local DP setting, the curator does not have access to the true data and can only see proxies generated by the users. This makes it more challenging to design local DP algorithms that can provide valid results leading to greater problems in developing corresponding theoretical properties and providing further statistical inference.

Researchers often propose consistent estimators for the parameters of interest and derive the asymptotic normality. However, these estimators often involve nuisance parameters that are not trivial to obtain or estimate, making them difficult to deploy in real-world scenarios. To address this issue, \citep{shao2010self} developed the methodology of self-normalization for constructing confidence intervals. This method involves designing a statistic called the self-normalizer, which is proportional to the nuisance parameters, and making the original estimate a pivotal quantity by placing it and the self-normalizer in the numerator and the denominator, thereby canceling out the nuisance parameters and leading to an asymptotically pivotal distribution. This methodology provides a powerful tool for statistical inference under complex data, particularly in the context of LDP frameworks where obtaining accurate original data or consistently estimating nuisance parameters without an additional privacy budget is challenging.

Efficient computation is essential for the practicality of LDP algorithms, as large sample sizes are necessary to counteract the effects of local perturbations and achieve optimal performance. Meanwhile, online computation is another valuable attribute of LDP algorithms, as it reduces storage requirements and diminishes risks associated with information storage. Early attempts of introduce online computation to DP algorithms can be traced back to \citep{jain2012differentially}, where additive Gaussian noise was injected into the gradient to provide DP protection. Later, \citep{pmlr-v70-agarwal17a} gives an online linear optimization DP algorithm that with optimal regret bounds. The concept of online computation has also been incorporated into federated learning, as discussed by \citep{9069945}. More recently, \citep{lee2022fast} has facilitated online computation for a random scaling quantity using only the trajectory of stochastic optimization, effectively eliminating the need for past state storage and enhancing computational efficiency. In contrast to traditional studies on DP online algorithms, our emphasis is on harnessing online computation for convenience. Our theoretical analysis concentrates on the statistical properties of the proposed estimators, encompassing aspects such as consistency, asymptotic normality, and more.

In this paper, our contributions are listed as follows.
\begin{itemize}
    \item We propose a new LDP algorithm for population quantile estimation that does not require a trusted curator. Under some mild conditions, we derive the consistency and asymptotic normality of the proposed quantile estimator. 
    \item We construct the confidence interval of the population quantiles via self-normalization, which eliminates the need for estimating the asymptotic variance in the limiting distribution. Furthermore, this procedure can be implemented online without storing all past statuses.
    
  \item  We also discuss the optimality of the proposed algorithm. By combining it with the central limit theorem of GDP, we demonstrate that our algorithm for median estimation achieves the lower bound of asymptotic variance among all median estimators constructed by a binary random response-based sequential interactive mechanism under LDP. 
\end{itemize}

The structure of this paper is as follows. We begin by providing an overview of the concepts of central DP and LDP. Then present our proposed methodology, detailing the algorithms and their corresponding theoretical results. Finally, we provide experimental results to demonstrate the effectiveness of our approach.

\section{Preliminaries}
\subsection{Central Differential Privacy}
\begin{definition}\citep{dwork2006our} 
A randomized algorithm $\mathcal{A}$, taking a dataset consisting of individuals as its input, is $(\epsilon, \delta)$-differentially private if, for any pair of datasets $S$ and $S^{\prime}$ that differ in the record of a single individual and any event $E$, satisfies the below condition:
\begin{equation*}
    \mathbb{P}[\mathcal{A}(S) \in E] \leq e^{\epsilon} \mathbb{P}\left[\mathcal{A}\left(S^{\prime}\right) \in E\right]+\delta.
\end{equation*}

When $\delta=0$, $\mathcal{A}$ is called 
$\epsilon$-Differentially Private ($\epsilon$-DP).
\end{definition}

The concept of DP only imposes constraints on the output distribution of an algorithm $\mathcal{A}$, rather than placing restrictions on the credibility of the entity running the algorithm or protecting the internal states of $\mathcal{A}$. The existence of the curator who has access to the raw data set is why this approach is known as "Central" DP. The curator simplifies the algorithm design and often leads to an asymptotically negligible
loss of accuracy from privacy protection \citep{cai2021cost}.
\subsection{Local Differential Privacy }
Despite the varying definitions of LDP due to the level of interactions, all of them depend on the following concept called $(\epsilon, \delta)$-randomizer. 
\begin{definition}\cite{8948625} An $(\epsilon, \delta)$-randomizer $R: X \rightarrow Y$ is an $(\epsilon, \delta)$-differentially private function taking a single data point as input.

\end{definition}

The definition of randomizer is mathematically a special case of the central DP. The main difference between the central and local DP is the role of the curator, which is further determined by the level of interactions allowed. In LDP, the curator coordinates interactions between $n$ users, each of whom holds their own private information $X_i$. In each round of interaction, the curator selects a user and assigns them a randomizer $R_t$. If the $(\epsilon, \delta)$ parameters are allowed by the experiment setting, the user will run the randomizer on their private information and release the output to the curator. 

The level of interactions can vary from full-interactive, where the curator can choose the randomizer and the next user based on all previous interactions, to sequential (also called one-shot) interactive, where the curator is not allowed to pick one user twice but is still able to adaptively picking the next the user-randomizer pairs based on all previous interactions, to non-interactive, where adaptivity is forbidden, and all user-randomizer pairs must be determined before any information is collected.
If the curator is further forbidden from varying the randomizer $R$ and tracking back outputs to a specific user, it will lead to another interesting setting called shuffle-DP \citep{cheu2019distributed}.

\subsection{Notations}
\par In this paper, we employ the following notations. $\mathbf{1}_{\left\{\cdot\right\}}$ is the indicator function and $[a]$ denotes the largest integer that does not exceed $a$. $\mathcal{O}$ (or ${\scriptstyle{\mathcal{O}}}%
$) denotes a sequence of real numbers of a certain order. For instance, ${\scriptstyle{\mathcal{O}}}(n^{-1/2})$ means
a smaller order than $n^{-1/2}$, and by $\mathcal{O}_{a.s.}$
(or ${\scriptstyle{\mathcal{O}}}_{a.s.}$) almost surely $\mathcal{O}$ (or ${\scriptstyle{\mathcal{O}}}$). For sequences $a_n$ and $b_n$, denote  $a_n \asymp b_n$ if there exist positive contants $c$ and $C$ such that $cb_n\leq a_n\leq Cb_n$. The symbol $\xrightarrow{d}$ means weak convergence or converge in distribution.

\section{Algorithm and Main Results}
\subsection{Algorithm}
\par Let $x_1,\dots, x_n, \dots$ be independently and identically distributed(i.i.d.) random variables defined on $\mathbb{R}$ representing private information of each user, with target quantile $\tau$ and corresponding true value $Q$, i.e., $\mathbb{P}(x_i \leq Q) = \tau$. To ensure the uniqueness of quantiles, we assume the $x_i$'s are continuous random variables, with positive density on the target quantile. In practice, we can perturb the data by a small amount of additive data-independent noise to remove atoms in the distribution as is in \citep{pmlr-v139-gillenwater21a}.

The design of the local randomizer is crucial for LDP mechanisms as it must properly choose the inquiry to the user in order to maximize the gathering of information related to the estimation of the target quantile without violating privacy conditions. The population quantiles can be considered as a minimizer of the check loss function: 
\begin{equation*}
    l_\tau(x,\theta)= \begin{cases}\tau (x-\theta) & \text { if } x \geq \theta \\ (\tau-1) (x-\theta), & \text { if } x<\theta\end{cases}.
\end{equation*}

In the non-DP case, a known solution is the use of stochastic gradient descent, as outlined in \citep{joseph2015stochastic}. It is important to note that for each point, the gradient it contributes is purely determined by the binary variable representing whether the value is greater than $\theta$ or not. 
This motivates us to modify the stochastic gradient descent process by adding a local randomization process, resulting in the Algorithm \ref{alg:LRC} and \ref{alg:main} outlined below:
\begin{algorithm}
   \caption{Locally Randomized Compare (LRC)}
   \label{alg:LRC}
\begin{algorithmic}
   \STATE {\bfseries Input:} Inquiry $q$, response rate $r$, private data $x$
   \STATE $u \sim Bernoulli(r)$
    \STATE $v \sim Bernoulli(0.5)$ 
   \IF{$u=1$}
   \STATE \textbf{return} $\mathbf{1}_{x>q}$
   \ELSE
   \STATE \textbf{return} $v$
   \ENDIF
\end{algorithmic}
\end{algorithm}

\begin{algorithm}
 \caption{Main Algorithm}
   \label{alg:main}
\begin{algorithmic}
   \STATE {\bfseries Input:} Step sizes $d_n$, target quantile $\tau \in (0,1)$, truthful response rate $r$
   \STATE Initialize: $n \gets 0$, $q_0 \gets 0$, $v^a_0\gets0$, $v^b_0\gets0$, $Q_0 \gets 0$
   \REPEAT
\STATE $n \gets n+1$
   \STATE Inquire: $s\gets LRC(q_{n-1},r,x_n)$
\IF{$s$ is $1$}
\STATE $$q_{n}\gets q_{n-1} +\frac{1-r+2\tau r}{2}d_n$$
\ELSE
\STATE $$q_{n}\gets q_{n-1} -\frac{1+r-2\tau r}{2}d_n$$
   \ENDIF
\STATE $Q_n= ((n-1)Q_n+q_n)/n$
\STATE $v^a_n\gets v^a_{n-1} +n^2 Q^2_n$
\STATE $v^b_n\gets v^b_{n-1} +n^2 Q_n$
\STATE Destroy $v^a_{n-1},v^b_{n-1},Q_{n-1},q_{n-1}$
\UNTIL{ Forever}
\end{algorithmic}
\end{algorithm}
In Algorithm \ref{alg:LRC}, generating randomness of $v$ before the if-condition fork may seem wasteful, but it prevents side-channel attacks such as inferring the true value based on the timing of response \cite{5207636,5370703}. Algorithm \ref{alg:main} collects random responses and generates the next inquiry accordingly. Therefore, Algorithm \ref{alg:main} satisfies the definition of sequential interactive local DP.

The following algorithm can be used when estimations and confidence intervals are required. These values are not calculated at every step to minimize computational expenses.
\begin{algorithm}[H]
 \caption{Generate Confidence Interval}
   \label{alg:infer}
\begin{algorithmic}
   \STATE {\bfseries Input:} Internal states of Algorithm \ref{alg:main}: $n$, $Q_n$, $v^a_n$, $v^b_n$
   \STATE $N_n\gets n^{-1}\left(v^a_n-2Q_n v^b_n+Q_n^2 n(n+1)(2 n+1)/6\right)$
   \STATE $W\gets n^{-1} \mathcal{U}_{1-\alpha/2}\sqrt{N_n}$ 
      \STATE {\bfseries Return:} Confidence interval $(Q_n-W,Q_n+W)$
\end{algorithmic}
\end{algorithm}

The use of dichotomous inquiry in data privacy brings multiple advantages. One benefit is the reduced communication cost, as it only takes one bit to respond.  Additionally, the binary response can make full use of the DP budget, as opposed to methods such as the Laplace mechanism, which may provide unnecessary privacy guarantees beyond $\epsilon$-DP,  
 as outlined in Theorem 3 in \cite{balle2018privacy} and Theorem 2.1 in \cite{IAMA}. 
 
 Furthermore, people tend to be more comfortable answering dichotomous questions compared to open-ended ones \citep{brown1996response} 
 as they present a choice between two options and may be perceived as less threatening than open-ended questions, which require more detailed and nuanced responses. In addition, the binary approach is easy to understand for users. With the proper choice of truthful response rate $r$, the algorithm known as the random response can be easily simulated through coin flips or dice rolls, allowing users to understand it fully and are able to "run" it without the help of electronic devices. This is in contrast to a DP mechanism involving the usage of random distribution on real numbers. Due to the finite nature of the computer, the imperfection of floating-point arithmetic leads to serious risks with effective exploits. For more information, please refer to \citep{mironov2012significance,jin2021we,haney2022precision}.

Before discussing the specific characteristics of our estimator, we will first demonstrate its performance through a sample trajectory. The experiment is conducted with a truthful response rate with $r = 0.5$, which means half of the responses are purely random. The objective is to estimate the median from i.i.d. samples. The true underlying distribution is a standard normal distribution.

It can be seen that from Figure \ref{fig:estimatetrend}, the proposed estimator converges to the true value, and both infeasible and proposed confidence intervals, defined later, contain the true value at a slightly larger sample size. Also, the proposed confidence intervals are highly competitive with the infeasible one in width. Refer to Figure \ref{fig:alterfigure1} and \ref{fig:alterfigure2} for convergence trajectories under different initialization or target quantiles.

\begin{figure}[H]
    \centering
    \includegraphics[width= 9 cm,height=6.5cm]{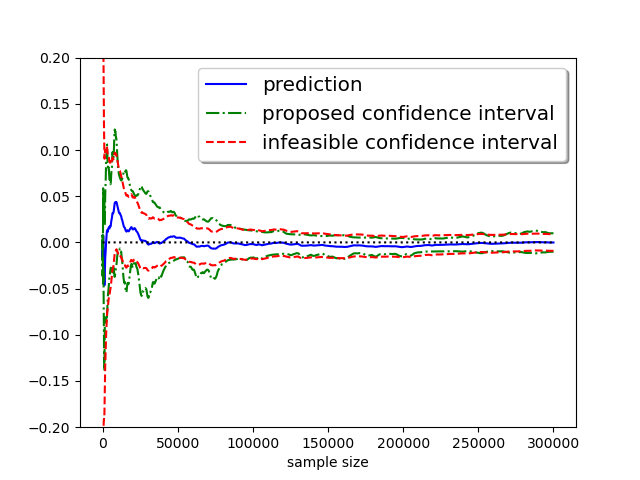}
    \caption{A sample trajectory of estimator $Q_n$, infeasible confidence interval (\ref{EQ:infeasibleinterval}) and proposed confidence interval (\ref{SNinerval}). The horizon dotted line is the true value $Q = 0$.   }
    \label{fig:estimatetrend}
\end{figure}

Next, we show the LDP property of our algorithm:

\begin{theorem}\label{THM:LRCDP}
    Algorithm \ref{alg:LRC} is an $(\epsilon,0)$-randomizer with $\epsilon=\log((1+r)/(1 - r)) $.
\end{theorem}
\begin{proof}
see Appendix \ref{proof:LRC}
\end{proof}

The algorithm presented in Algorithm \ref{alg:main} adaptively selects the next randomizer, determined by the parameter $q$ in Algorithm \ref{alg:LRC}, based on its internal state $q_n$. However, it never revisits previous users. As a result, Algorithm \ref{alg:main} satisfies sequential interactive $(\epsilon,0)$-LDP, where $\epsilon=\log\left((1+r)/(1 - r)\right)$ (equivalently, $r=(e^\epsilon-1)/(e^\epsilon+1)=\tanh\left(\epsilon/2\right)$).

Throughout the remainder of this paper, we will use the truthful response rate $r$ to represent the privacy budget, as opposed to the more standard $\epsilon$. This choice is made for the following reasons:

In the context of LDP, it is crucial to ensure understanding and acceptance by end-users who may not possess expertise in the field. The truthful response rate, denoted by $r$, has a more intuitive interpretation. Additionally, $r$ appears in multiple results presented in this paper, and maintaining this form allows for a more direct presentation. If necessary, the results can be easily converted by replacing all instances of $r$ with $\tanh\left(\frac{\epsilon}{2}\right)$. For a conversion table, please refer to Table \ref{reconversiontable}.

\subsection{Consistency}
\par To discuss the asymptotic properties of estimator $Q_n$, we rewrite it as a recursive equation. Let $\{U_n\}$ and $\{V_n\}$ be the i.i.d. Bernoulli sequences with
\begin{align*}
  &\mathbb{P}(U_n = 1) = r,~\mathbb{P}(U_n = 0) = 1-r,\\
  & \mathbb{P}(V_n = 1) = \mathbb{P}(V_n = 0) = 1/2.
\end{align*}
For $q_0 \in \mathbb{R}$,
\begin{equation}
\begin{split}
\label{EQ:Qn}
    q_{n+1} &= q_n +d_n\frac{1-r+2r\tau}{2}\left(\mathbf{1}_{x_{n+1}>q_n}U_n + (1-U_n)V_n\right) \\
    &- d_n\frac{1+r-2r\tau}{2}\left(\mathbf{1}_{x_{n+1}<q_n}U_n + (1-U_n)(1-V_n)\right),
\end{split}
\end{equation}
where the step size $\{d_n\}_{n = 1}^{\infty}$ , satisfies
\begin{equation*}
    \sum_{n = 1}^{\infty}d_n = \infty,\qquad \sum_{n = 1}^{\infty}d_n^2 <\infty.
\end{equation*}
The step size $d_n$ is vital for the convergence of $q_n$, but it has a relatively minor effect on $Q_n$. The following theorem guarantees consistency:
\begin{theorem}\label{THM:consist}
\par For increasing positive number $\gamma_n$, satisfied
\begin{equation*}
    \frac{\gamma_n}{\gamma_{n-1}} = 1+ {\scriptstyle{\mathcal{O}}}(d_n),\qquad \sum_{n = 1}^{\infty}d_n^2\gamma_n^2 <\infty,
\end{equation*}
the $n$-step output $q_n$ enjoys that
\begin{equation*}
    \gamma_n\left|q_n-Q\right| = {\scriptstyle{\mathcal{O}}}_{a.s.}(1).
\end{equation*}
\end{theorem}
\begin{proof}
    see Appendix \ref{proof:asym}.
\end{proof}
\par In particular, if $d_n \asymp a/n^{\beta}$, for some constant $a>0$ and $\beta \in (1/2,1)$, then $\gamma_n \asymp n^{\gamma}$ for some $\gamma< \beta - 1/2$, and for the sake of simplicity, we will set the step sizes as $d_n \asymp a/n^{\beta}$. 
\subsection{Asymptotic Normality}
Next, the asymptotic normality will be discussed.
\begin{theorem}\label{THM:normality}
\par If $\beta \in (0,1)$, then
\begin{equation*}
    \sqrt{n}\left(Q_n - Q\right)\xrightarrow{d} N \left(0, \frac{1 - r^2(1-2(1-\tau))^2}{4r^2f_{X}^2(Q)}\right),
\end{equation*}
where $f_X(Q)$ is the value on $Q$ for density function of $X$. 
\end{theorem}
\begin{proof}
    see Appendix \ref{proof:asym}.
\end{proof}
\par Noticed that the conditions on $\beta$ in Theorem \ref{THM:consist} and Theorem \ref{THM:normality} are different. It is possible that $q_n$ fails to converge to $Q$, but $Q_n$ still enjoys asymptotic normality. Following Theorem \ref{THM:normality}, one constructs the confidence interval of $Q$, if $f_X(Q)$ can be obtained or estimated by $\widehat{f_X(Q)}$. Denote $z_{1-\alpha}$ as the upper $\alpha-$quantile of standard normal distribution. The infeasible confidence interval with significance level $\alpha$ is:
\begin{equation}
\label{EQ:infeasibleinterval}
\begin{split}
 &\left(Q_n - z_{1-\alpha}\sqrt{n(1 - r^2(1-2(1-\tau))^2)}/(2r\widehat{f_X(Q)}), \right.\\
    &\left.Q_n + z_{1-\alpha}\sqrt{n(1 - r^2(1-2(1-\tau))^2)}/(2r\widehat{f_X(Q)})\right).   
\end{split}
\end{equation}
\par However, obtaining a consistent estimator $\widehat{f_X(Q)}$, such as using non-parametric methods under our differential privacy framework, is not straightforward, since we can only obtain the binary sequence $\mathbf{1}_{x_n > q_{n-1}}$ for protecting privacy, and the original data set $x_1,\dots, x_n$ cannot be accessed directly. 

\par An alternative approach to estimate the nuisance parameter $f_X(Q)$ is through the use of bootstrap methods to simulate the asymptotic distribution. Traditional bootstrap methods that rely on re-sampling are not suitable for the stochastic gradient descent method because of failing to recover the special dependence structure defined in (\ref{EQ:Qn}). 

\par Recently, \citep{fang2018online} proposed online bootstrap confidence intervals for stochastic gradient descent, which involve recursively updating randomly perturbed stochastic estimates. Although this approach performs well when there are no constraints on DP, it requires multiple interactions with the users and will therefore blow up the privacy budget.

\subsection{Inference via Self-normalization}
\par To overcome the difficulties above, we propose a novel inference procedure of quantiles under the LDP framework via self-normalization, which will avoid estimating the nuisance parameter $f_X(Q)$. We hope to construct an estimator that is proportional to the nuisance parameters.  To approach that, we will first establish further theoretical properties of the proposed estimator $Q_n$. Define the process $S_{[nt]} = \sum_{i = 1}^{[nt]}q_i$, $t \in [0,1]$.
\begin{theorem}\label{THM:weakconvergence}
\par If  $\beta \in (0,1)$, then
\begin{equation*}
    n^{-1/2}(S_{[nt]}-nQ) \xrightarrow{d} \frac{\sqrt{(1 - r^2(1-2(1-\tau))^2)}}{2rf_{X}(Q)} W(t),
\end{equation*}
where $W(t)$ is the Brownian motion in $(C[0,1],\mathbb{R})$.
\end{theorem}
\begin{proof}
    see Appendix \ref{proof:asym}.
\end{proof}
\par Noticed that Theorem \ref{THM:normality} is the special case in Theorem \ref{THM:weakconvergence} when $t= 1$. Then, following Theorem \ref{THM:weakconvergence}, we define the self-normalizer: 
\begin{equation*}
     N_n = \int_{0}^{1}
    \left(S_{[nt]} - [nt]Q_n\right)^2dt,
\end{equation*}
By the continuous mapping theorem, we can derive:
\begin{align*}
    \frac{n^{-1/2}(S_{n}- nQ)}{\sqrt{n^{-1}N_n}} \xrightarrow{d} \mathcal{S}: = \frac{W(1)}{\sqrt{\int_0^1\left(W(t) - tW(1)\right)^2dt}}, 
\end{align*}
where the asymptotical distribution $\mathcal{S}$ is not associated with any unknown parameters, and its quantile can be computed by Monte Carlo simulation. Therefore, we have constructed an asymptotical pivotal quantity.  Denote $\mathcal{U}_{1-\alpha}$ the $1-\alpha$ quantile of $\mathcal{S}$, the $1-\alpha$ self-normalized confidence interval of $Q$ is constructed by:
\begin{align}\label{SNinerval}
    \left(Q_n - n^{-1}\mathcal{U}_{1-\alpha/2}\sqrt{N_n}, Q_n + n^{-1}\mathcal{U}_{1-\alpha/2}\sqrt{N_n}\right).
\end{align}
\par As noted by \citep{shao2015self}, the distribution of $\mathcal{S}$ has a heavier tail than that of the standard normal distribution, which is analogous to the heavier tail of $t-$distribution compared to the standard normal distribution, resulting in a wider but not conservative corresponding confidence interval. However, the average width of the confidence interval constructed through self-normalization is not excessively large when compared to the infeasible confidence interval, as demonstrated by numerical experiments in Figure  \ref{fig:estimatetrend}. Furthermore, the construction of an asymptotic pivotal quantity is not unique. See Appendix \ref{alternativeselfnorm} for examples of other possibilities.

Whether there are theoretical advantages between the different constructions of self-normalizer is still open to discussion, but according to \citep{lee2022fast}, the proposed self-normalizer can be computed in a fully online fashion and is computationally efficient, as outlined in Algorithm \ref{alg:main} and \ref{alg:infer}. The algorithm only needs to store
a single integer $n$ and four float numbers: $ v_n^a, v_n^b, q_n, Q_n$ and conduct only a dozen of arithmetic operations.

\subsection{Discussion of Optimality}
\par In this subsection, we will discuss the optimality of the proposed algorithm. To generalize the setting, we consider all binary random response-based sequential interactive mechanisms.
The random response mechanism can be written as the following $K:\{0,1\}\rightarrow\{0,1\}$:

$$K(x) =       \begin{cases}
                          0, &  \text{w.p.}~~ (1-r)/2, \\
                          1, &  \text{w.p.}~~ (1-r)/2,
\\           x, &  \text{w.p.}~~r .
                        \end{cases}$$

Let $\{T_1,\cdots,T_n\}$ be a collection of binary query functions, which means $T_i(x)=\mathbf{1}_{x\in C_i}$, for some subset $C\subset \mathbb{R}$. In the sequential interactive LDP setting, the curator will generate its output based on the transcript $\{\{K\circ T_1(x_1),\cdots,K\circ  T_n(x_n)\},\{C_1,\cdots,C_n\}\}$ and the choice of $C_i$ may depend on the transcript up to this point:$\{\{K\circ T_1(x_1),\cdots,K\circ  T_{i-1}(x_{i-1})\},\{C_1,\cdots,C_{i-1}\}\}$. Notice that the Algorithm \ref{alg:LRC} is a special case where $C_i  = \{z: z\geq q_{i-1}\}$, and $q_{i-1}$ is given by 
\begin{align*}
  \sum_{j=1}^{i-1}T_{j}(x_j)\frac{1-r+2\tau r}{2}d_j-(1-T_{j}(x_j))\frac{1+r-2\tau r}{2}d_j.
\end{align*}

We aim to determine a lower bound for the estimation variance. Therefore, any lower bounds derived under specific conditions also serve as a general lower bound for the estimation variance. To demonstrate this, we will present a pair of distributions with distinct medians that are, to the best of our knowledge, the most indistinguishable given randomized binary queries.

Define:
\begin{align}\label{EQ:hypotheses0}
    H_0: x_i \sim Laplace(1)  \mbox{ vs. } H_1: x_i \sim Laplace(1)+\epsilon_n
\end{align}
Let $\epsilon=\log\left[(e^{\frac{1}{\sqrt{n}}} (r+1)+r-1)/(e^{\frac{1}{\sqrt{n}}} (r-1)+r+1)\right]$. Simple computation yields that for any $(a,b)\in\{0,1\}^2$

\begin{equation}\label{EQ:KTDP2}
    \frac{\mathbb{P}(K\circ T_i(x_i)=a|H_b)}{\mathbb{P}(K\circ T_i(x_i)=a|H_{1-b})}\leq \frac{e^{\epsilon_n} (r+1)-r+1}{-e^{\epsilon_n} (r-1)+r+1}=e^{\sqrt{\frac{1}{n}}}.
\end{equation}

Interestingly, if we consider the truth $H\in\{H_0,H_1\}$ as a data set containing only one data point,  (\ref{EQ:KTDP2}) shows that $K\circ T_i$ is $1/\sqrt{n}$-DP. Notice that the transcript is a $n$-fold adaptive composition \citep{kairouz2015composition} of $1/\sqrt{n}$-DP mechanisms.
By Theorem 8 \citep{dong2019gaussian}, the transcript and all post-processing of it (Proposition 4; \citep{dong2019gaussian}) asymptotically satisfies the Gaussian Differential Privacy condition with $\mu=1$ (or briefly $1$-GDP).

We will now examine the limit on the best possible variance imposed by the $1$-GDP condition. Denote the estimator of median as $\hat{\theta}_n$. First, we will consider asymptotically normal, unbiased, shift-invariant estimators of the median. By restricting our discussion to unbiased, shift-invariant estimators, we ensure that no estimator has an unfair advantage by favoring specific values. Under the null hypothesis, for the standard deviation $\sigma_n$ of $\hat{\theta}_n$, one has that
\begin{equation*}
\frac{\hat{\theta}_n}{\sigma_n}\xrightarrow{d} N(0,1),  
\end{equation*}
and under the alternative hypothesis,
\begin{equation*}
    \frac{\hat{\theta}_n-\epsilon_n}{\sigma_n}\xrightarrow{d}
    N(0,1).
\end{equation*}

The $1$-GDP condition implies that for sufficiently large $n$, $\epsilon_n/\sigma_n\leq 1$ (
see Appendix \ref{APfinal}). By plugging in the values $\epsilon_n= (r\sqrt{n})^{-1}+\mathcal{O}(n^{-3/2})$ and $1/2=f(F^{-1}(1/2))$, we deduce that:
\begin{equation*}
    \sigma_n\geq \frac{1}{2 r\sqrt{n}f(F^{-1}(1/2))}+\mathcal{O}\left(n^{-1}\right),
\end{equation*}
which gives us an asymptotic lower bound of the variance: $(4 r^2nf^2(F^{-1}(1/2)))^{-1}$. This lower bound matches the asymptotic variance obtained in Theorem \ref{THM:normality}, showing the optimality of our approach. Although most estimators we are interested in have an asymptotically normal distribution, we wish to generalize the minimal variance result to other families as the theorem below.

\begin{theorem}\label{THM:optimal}
\par If $\hat{\theta}_n$ is a median estimator based on the random response of binary-based sequential interactive inquiries such that:

\begin{equation*}
\frac{\hat{\theta}_n-F^{-1}(1/2)}{\sigma_n}
\xrightarrow{d} G
\end{equation*}
where $G$ has a log-concave density $f_G(x) \propto \mathrm{e}^{-\varphi(x)}$ on $\mathbb{R}$ such that $\varphi(x)=\varphi(-x)$, $\mathbb{E}\left[\left(\varphi'(G)\right)^2\right]<+\infty$,  and $\mathbb{E}\left[G^2\right]=1$.

Then,
\begin{equation*}
    \sigma_n\geq \frac{1}{2 r\sqrt{n}f(F^{-1}(1/2))}+\mathcal{O}\left(n^{-1}\right).
\end{equation*}
\end{theorem}

The minimal variance result can be attributed to two factors. In Appendix \ref{APfinal}, we demonstrate that asymptotic GDP imposes a condition on the variance of estimators that follow a normal distribution. This condition serves as a lower bound for $1$-GDP estimators, without relying on any specific mechanism assumption. Secondly, the relaxation from the assumption of normality to milder conditions on the function $G$ is a consequence of Theorem 1.2 in \citep{NEURIPS2021_7c2c48a3}. This theorem establishes that among all $\mu$-GDP estimators satisfying the aforementioned conditions, the variance is lower bounded by $1/\mu^2$. This lower bound is attainable when the underlying distribution is normal.

\section{Experiments} \label{sec:experiments}
We evaluate the performance of our algorithms using a variety of distributions. The data come from four cases: standard Normal $N(0,1)$, Uniform $U(-1,1)$, standard Cauchy $C(0,1)$, and PERT distribution \citep{clark1962pert} with probability density function:
\begin{equation*}
f(x) = 0.625(1-x)(1+x)^3, ~~ x\in (-1,1).
\end{equation*}
These cases represent situations with heavy tails, compact or non-compact support, and asymmetric distributions commonly found in practice, as shown in Figure \ref{fig:densiyu}.

\begin{figure}[h!]
    \centering
    \includegraphics[width= 7.5cm,height=5.5cm]{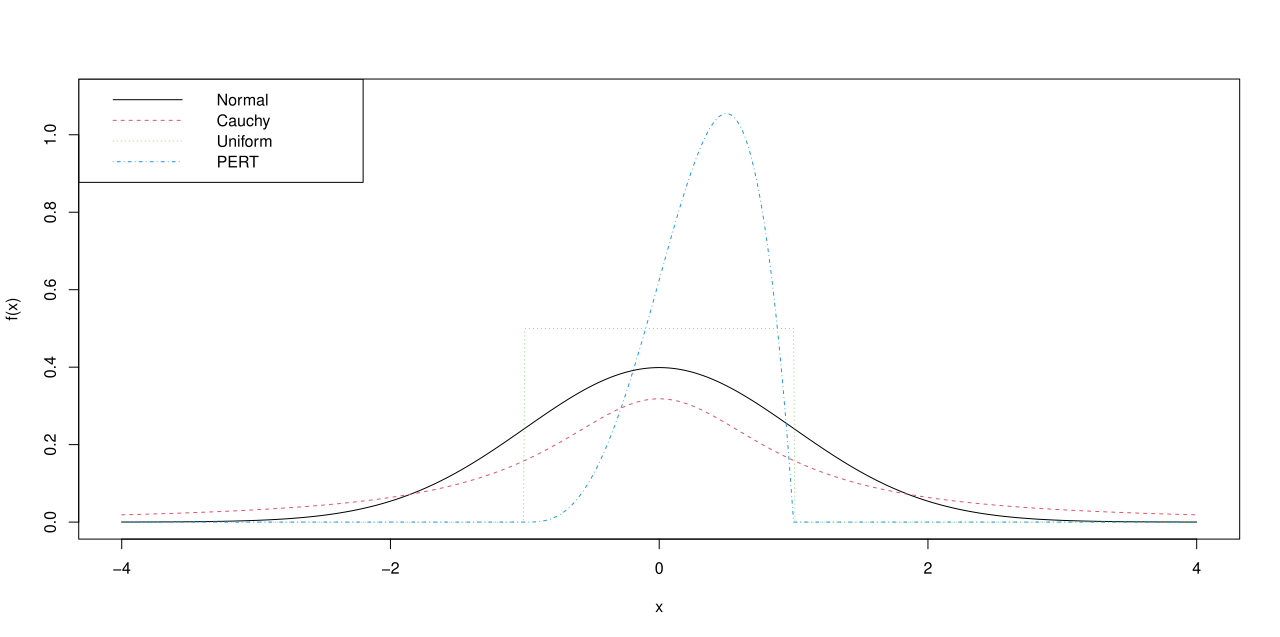}
    \caption{Plot of the density function, where the types of lines represent different distribution, solid: Normal, dashed: Cauchy, dotted: Uniform, dot-dash: PERT.  }
    \label{fig:densiyu}
\end{figure}

The target quantiles are $\tau = 0.3, 0.5, 0.8$, and the truthful response rate $r = 0.25, 0.5, 0.9$, which  the privacy budget is 
$\epsilon = \log(1+2r/(1-r))$ corresponding to $0.51, 1.09, 2.94$ respectively. We use the step sizes  $d_n = 2/(n^{0.51}+100)$ for all experiments, which satisfies the assumptions of Theorem \ref{THM:normality} and \ref{THM:weakconvergence}. The range of sample size $n$ is $(10000,400000)$, the initial value $q_0 = 0$, and the number of replication is $10000$. The results from different sample sizes are independently conducted from scratch to eliminate the correlation among experiments.  

To show the consistency of the proposed estimator $Q_n$, Figure \ref{Figure2} displays the box plots of estimator $Q_n$ under Normal distribution with sample size $n = 10000,\dots,50000$. As the sample size increases, the estimation becomes closer to the true values $Q$, the corresponding standard errors decay across all settings, and the truthful response rate leads to significantly better performance in small finite sample sizes but has diminishing effects afterward. Meanwhile,  we can also see that the proximity between the true target value and the initialization $0$ is beneficial to early performances. But in an asymptotic view, the proposed algorithm is insensitive to the initial value selection.

\par We also demonstrate the empirical coverage rate and mean absolute error of the developed method in Table \ref{tab:normal}. The empirical coverage rate of the proposed method becomes closer to the nominal confidence level as the sample size increases in most cases and the mean absolute error decreases to zero. The corresponding figures and tables of other distributions can be found in Appendix \ref{Sec:c}, which describes a similar phenomenon. 

\par Figure \ref{fig:QQplot} investigates the performance of the proposed confidence interval in other nominal levels. One can discover that the curves of the empirical coverage rate are getting closer to $y = x$ uniformly, as sample size increases in all privacy budget settings, which reveals the performance of the proposed method is irrelevant to the pre-determined significance level. It is worth noting that when $r = 0.25$, the effective sample size is $1/16$ of the original one, yet the performance of the proposed method remains excellent, which strongly supports the asymptotic theory.

\begin{figure}[h!]
  \centering  \includegraphics[width= 8.4cm,height=8.4cm]{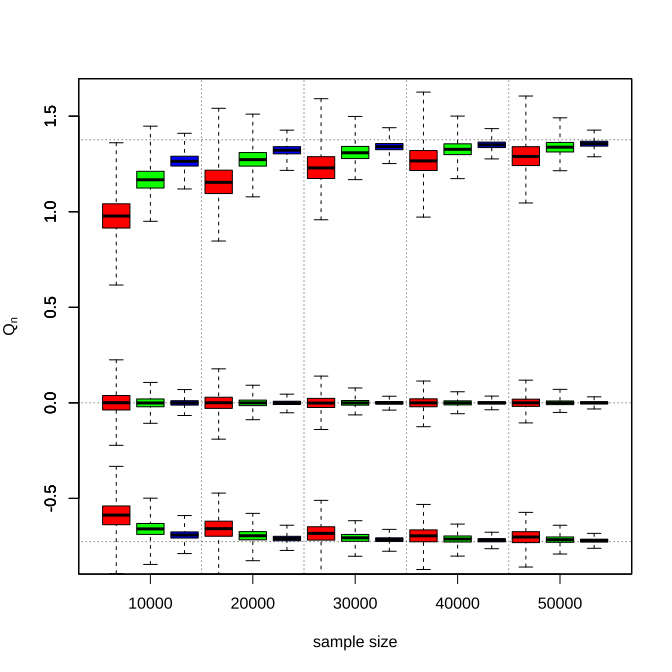}
  \caption{Box-plot of estimator $Q_n$ for different target quantiles of Normal distribution. In each sample size divided by a vertical dotted line, the three boxes establish results with different privacy budgets by left: $r = 0.25$, middle $r = 0.5$, and right: $r = 0.9$. The horizontal dashed lines represent the true value $Q$ in $\tau = 0.3$, $0.5$, $0.8$ from the bottom to the top.}\label{Figure2}
\end{figure}

\begin{figure}
    \centering
    \subfigure[Left: $n = 10000$. Right: $n = 50000$]{
    \includegraphics[width= 4.3 cm,height=3.1cm]{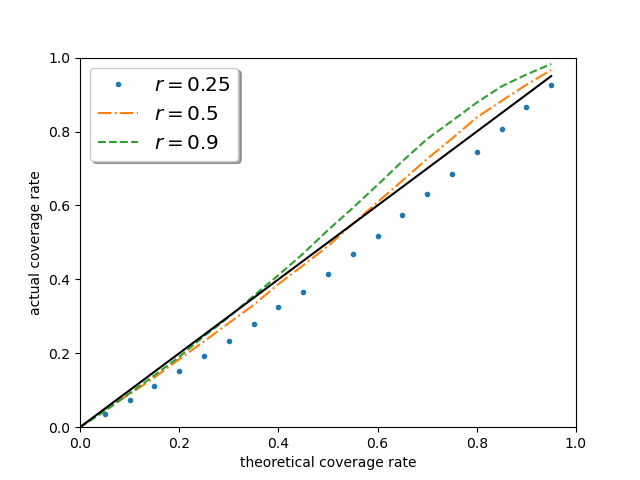}
    \includegraphics[width= 4.3 cm,height=3.1cm]{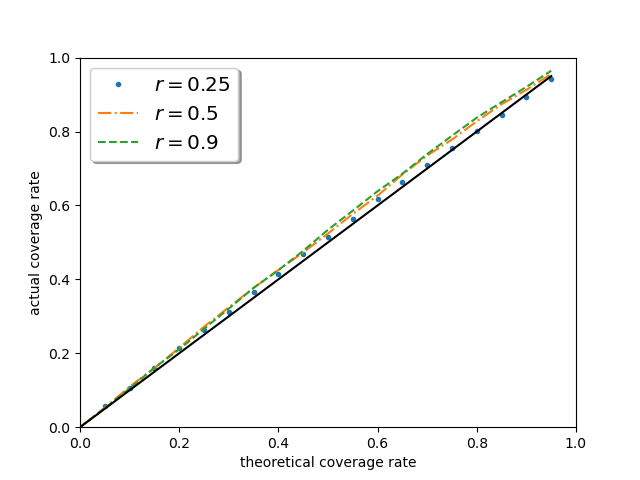}
    }
    \subfigure[Left: $n = 100000$. Right: $n = 200000$]{
    \includegraphics[width= 4.3 cm,height=3.1cm]{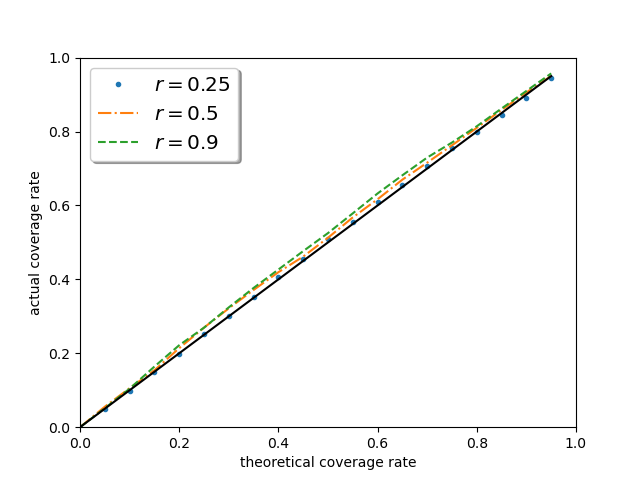}
    \includegraphics[width= 4.3 cm,height=3.1cm]{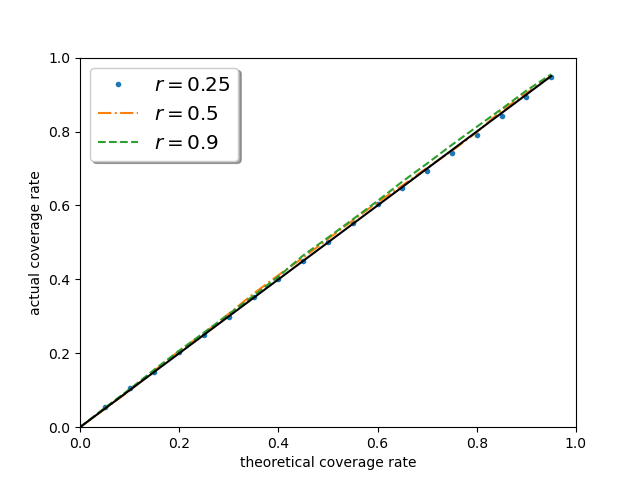}
    }
    \caption{The curve of the empirical coverage rate of proposed confidence interval (\ref{SNinerval}) with nominal significance level, when the data are Normal and target quantile $\tau = 0.3$ under different privacy budget (dotted $r = 0.25$, dot-dash $r=0.5$ and dashed $r =0.9$).}
    \label{fig:QQplot}
\end{figure}

\begin{table}[h!]
  \centering
  \caption{Empirical results of coverage rate(mean absolute error) of proposed confidence interval (\ref{SNinerval}) (estimator $Q_n$) with data collected from Normal.}
  \resizebox{85mm}{!}{
    \begin{tabular}{ccccc}\toprule
          $n$   & $\tau$ & $r = 0.25$ & $r=0.5$ & $r=0.9$ \\\hline
    \multirow{3}[0]{*}{10000} & 0.3   & 0.926(0.069) & 0.965(0.034) & 0.982(0.018) \\
          & 0.5   & 0.834(0.037) & 0.897(0.019) & 0.911(0.011) \\
          & 0.8   & 0.962(0.121) & 0.992(0.058) & 0.999(0.031) \\\hline
    \multirow{3}[0]{*}{20000} & 0.3   & 0.936(0.041) & 0.958(0.020) & 0.971(0.011) \\
          & 0.5   & 0.888(0.027) & 0.915(0.014) & 0.936(0.008) \\
          & 0.8   & 0.965(0.063) & 0.984(0.030) & 0.994(0.016) \\\hline
    \multirow{3}[0]{*}{40000} & 0.3   & 0.943(0.025) & 0.958(0.013) & 0.967(0.007) \\
          & 0.5   & 0.910(0.020) & 0.931(0.010) & 0.937(0.006) \\
          & 0.8   & 0.966(0.035) & 0.978(0.017) & 0.984(0.009) \\\hline
    \multirow{3}[0]{*}{100000} & 0.3   & 0.946(0.015) & 0.954(0.007) & 0.958(0.004) \\
          & 0.5   & 0.929(0.013) & 0.944(0.006) & 0.941(0.004) \\
          & 0.8   & 0.954(0.019) & 0.965(0.009) & 0.973(0.005) \\\hline
    \multirow{3}[0]{*}{200000} & 0.3   & 0.947(0.010) & 0.951(0.005) & 0.956(0.003) \\
          & 0.5   & 0.942(0.009) & 0.949(0.004) & 0.947(0.002) \\
          & 0.8   & 0.956(0.013) & 0.960(0.006) & 0.964(0.003) \\\hline
    \multirow{3}[0]{*}{400000} & 0.3   & 0.945(0.007) & 0.953(0.004) & 0.948(0.002) \\
          & 0.5   & 0.942(0.006) & 0.949(0.003) & 0.944(0.002) \\
          & 0.8   & 0.952(0.009) & 0.957(0.004) & 0.958(0.002) \\
\toprule
    \end{tabular}}
  \label{tab:normal}%
\end{table}%

\section{Conclusion and Future Works}

\par In this paper, we proposed a novel algorithm for
estimating population quantiles under the settings of LDP. The core design idea of the algorithm is based on using dichotomous inquiry. The proposed estimator enjoys excellent theoretical properties, including consistency, asymptotic normality, and optimality in some special cases. Importantly, by applying the technique of self-normalization to cancel out the nuisance parameters, we can construct confidence intervals of population quantiles for statistical inference. Finally, our algorithm is designed in an online setting, making it suitable for handling large streaming data without the need for data storage. Extensive simulation studies reveal a positive
confirmation of the asymptotic theory.

Despite the contributions above, this article still leaves many exciting questions unanswered, which opens many avenues for future research. A general tight lower bound for other quantiles under our setting is still undetermined, and we have yet to consider other variants of LDP (e.g. full-interactive). Other directions include exploring data that is not independently and identically distributed, such as time series or spatial series data. Additionally, the quantile of interest may be influenced by other covariates, leading to the study of LDP quantile regression. This paper focuses on estimating quantiles for a specific sample size $n$, with the potential for developing consistent bounds, resulting in the transition from quantile confidence intervals to confidence sequences. 

\nocite{langley00}

\bibliography{example_paper}

\begin{thebibliography}{38}
\providecommand{\natexlab}[1]{#1}
\providecommand{\url}[1]{\texttt{#1}}
\expandafter\ifx\csname urlstyle\endcsname\relax
  \providecommand{\doi}[1]{doi: #1}\else
  \providecommand{\doi}{doi: \begingroup \urlstyle{rm}\Url}\fi

\bibitem[Agarwal \& Singh(2017)Agarwal and Singh]{pmlr-v70-agarwal17a}
Agarwal, N. and Singh, K.
\newblock The price of differential privacy for online learning.
\newblock In Precup, D. and Teh, Y.~W. (eds.), \emph{Proceedings of the 34th
  International Conference on Machine Learning}, volume~70 of \emph{Proceedings
  of Machine Learning Research}, pp.\  32--40. PMLR, 06--11 Aug 2017.

\bibitem[Alabi et~al.(2022)Alabi, Ben-Eliezer, and
  Chaturvedi]{alabi2022bounded}
Alabi, D., Ben-Eliezer, O., and Chaturvedi, A.
\newblock Bounded space differentially private quantiles.
\newblock \emph{arXiv preprint arXiv:2201.03380}, 2022.

\bibitem[Amin et~al.(2020)Amin, Joseph, and Mao]{pmlr-v125-amin20a}
Amin, K., Joseph, M., and Mao, J.
\newblock Pan-private uniformity testing.
\newblock In Abernethy, J. and Agarwal, S. (eds.), \emph{Proceedings of Thirty
  Third Conference on Learning Theory}, volume 125 of \emph{Proceedings of
  Machine Learning Research}, pp.\  183--218. PMLR, 09--12 Jul 2020.

\bibitem[Balle et~al.(2018)Balle, Barthe, and Gaboardi]{balle2018privacy}
Balle, B., Barthe, G., and Gaboardi, M.
\newblock Privacy amplification by subsampling: Tight analyses via couplings
  and divergences.
\newblock \emph{Advances in Neural Information Processing Systems}, 31, 2018.

\bibitem[Ben-Eliezer et~al.(2022)Ben-Eliezer, Mikulincer, and
  Zadik]{ben2022archimedes}
Ben-Eliezer, O., Mikulincer, D., and Zadik, I.
\newblock Archimedes meets privacy: On privately estimating quantiles in high
  dimensions under minimal assumptions.
\newblock \emph{arXiv preprint arXiv:2208.07438}, 2022.

\bibitem[Brown et~al.(1996)Brown, Champ, Bishop, and
  McCollum]{brown1996response}
Brown, T.~C., Champ, P.~A., Bishop, R.~C., and McCollum, D.~W.
\newblock Which response format reveals the truth about donations to a public
  good?
\newblock \emph{Land Economics}, pp.\  152--166, 1996.

\bibitem[Cai et~al.(2021)Cai, Wang, and Zhang]{cai2021cost}
Cai, T.~T., Wang, Y., and Zhang, L.
\newblock The cost of privacy: Optimal rates of convergence for parameter
  estimation with differential privacy.
\newblock \emph{The Annals of Statistics}, 49\penalty0 (5):\penalty0
  2825--2850, 2021.

\bibitem[Cheu et~al.(2019)Cheu, Smith, Ullman, Zeber, and
  Zhilyaev]{cheu2019distributed}
Cheu, A., Smith, A., Ullman, J., Zeber, D., and Zhilyaev, M.
\newblock Distributed differential privacy via shuffling.
\newblock In \emph{Advances in Cryptology--EUROCRYPT 2019: 38th Annual
  International Conference on the Theory and Applications of Cryptographic
  Techniques, Darmstadt, Germany, May 19--23, 2019, Proceedings, Part I 38},
  pp.\  375--403. Springer, 2019.

\bibitem[Clark(1962)]{clark1962pert}
Clark, C.~E.
\newblock The pert model for the distribution of an activity time.
\newblock \emph{Operations Research}, 10\penalty0 (3):\penalty0 405--406, 1962.

\bibitem[Coppens et~al.(2009)Coppens, Verbauwhede, De~Bosschere, and
  De~Sutter]{5207636}
Coppens, B., Verbauwhede, I., De~Bosschere, K., and De~Sutter, B.
\newblock Practical mitigations for timing-based side-channel attacks on modern
  x86 processors.
\newblock In \emph{2009 30th IEEE Symposium on Security and Privacy}, pp.\
  45--60, 2009.
\newblock \doi{10.1109/SP.2009.19}.

\bibitem[Ding et~al.(2017)Ding, Kulkarni, and Yekhanin]{ding2017collecting}
Ding, B., Kulkarni, J., and Yekhanin, S.
\newblock Collecting telemetry data privately.
\newblock \emph{Advances in Neural Information Processing Systems}, 30, 2017.

\bibitem[Dippon(1998)]{dippon1998globally}
Dippon, J.
\newblock Globally convergent stochastic optimization with optimal asymptotic
  distribution.
\newblock \emph{Journal of applied probability}, 35\penalty0 (2):\penalty0
  395--406, 1998.

\bibitem[Dong et~al.(2021{\natexlab{a}})Dong, Roth, and Su]{dong2019gaussian}
Dong, J., Roth, A., and Su, W.~J.
\newblock Gaussian differential privacy.
\newblock \emph{{J}ournal of the {R}oyal {S}tatistical {S}ociety: Series B
  (Statistical Methodology)}, 2021{\natexlab{a}}.

\bibitem[Dong et~al.(2021{\natexlab{b}})Dong, Su, and
  Zhang]{NEURIPS2021_7c2c48a3}
Dong, J., Su, W., and Zhang, L.
\newblock A central limit theorem for differentially private query answering.
\newblock In Ranzato, M., Beygelzimer, A., Dauphin, Y., Liang, P., and Vaughan,
  J.~W. (eds.), \emph{Advances in Neural Information Processing Systems},
  volume~34, pp.\  14759--14770. Curran Associates, Inc., 2021{\natexlab{b}}.

\bibitem[Dwork \& Lei(2009)Dwork and Lei]{dwork2009differential}
Dwork, C. and Lei, J.
\newblock Differential privacy and robust statistics.
\newblock In \emph{Proceedings of the forty-first annual ACM symposium on
  Theory of computing}, pp.\  371--380, 2009.

\bibitem[Dwork et~al.(2006)Dwork, Kenthapadi, McSherry, Mironov, and
  Naor]{dwork2006our}
Dwork, C., Kenthapadi, K., McSherry, F., Mironov, I., and Naor, M.
\newblock Our data, ourselves: Privacy via distributed noise generation.
\newblock In \emph{{A}nnual {I}nternational {C}onference on the {T}heory and
  {A}pplications of {C}ryptographic {T}echniques}, pp.\  486--503. {S}pringer,
  2006.

\bibitem[Erlingsson et~al.(2014)Erlingsson, Pihur, and
  Korolova]{erlingsson2014rappor}
Erlingsson, {\'U}., Pihur, V., and Korolova, A.
\newblock Rappor: Randomized aggregatable privacy-preserving ordinal response.
\newblock In \emph{Proceedings of the 2014 ACM SIGSAC conference on computer
  and communications security}, pp.\  1054--1067, 2014.

\bibitem[Evfimievski et~al.(2003)Evfimievski, Gehrke, and
  Srikant]{evfimievski2003limiting}
Evfimievski, A., Gehrke, J., and Srikant, R.
\newblock Limiting privacy breaches in privacy preserving data mining.
\newblock In \emph{Proceedings of the twenty-second ACM SIGMOD-SIGACT-SIGART
  symposium on Principles of database systems}, pp.\  211--222, 2003.

\bibitem[Fang et~al.(2018)Fang, Xu, and Yang]{fang2018online}
Fang, Y., Xu, J., and Yang, L.
\newblock Online bootstrap confidence intervals for the stochastic gradient
  descent estimator.
\newblock \emph{The Journal of Machine Learning Research}, 19\penalty0
  (1):\penalty0 3053--3073, 2018.

\bibitem[Gillenwater et~al.(2021)Gillenwater, Joseph, and
  Kulesza]{pmlr-v139-gillenwater21a}
Gillenwater, J., Joseph, M., and Kulesza, A.
\newblock Differentially private quantiles.
\newblock In Meila, M. and Zhang, T. (eds.), \emph{Proceedings of the 38th
  International Conference on Machine Learning}, volume 139 of
  \emph{Proceedings of Machine Learning Research}, pp.\  3713--3722. PMLR,
  18--24 Jul 2021.

\bibitem[Haney et~al.(2022)Haney, Desfontaines, Hartman, Shrestha, and
  Hay]{haney2022precision}
Haney, S., Desfontaines, D., Hartman, L., Shrestha, R., and Hay, M.
\newblock Precision-based attacks and interval refining: how to break, then
  fix, differential privacy on finite computers.
\newblock \emph{arXiv preprint arXiv:2207.13793}, 2022.

\bibitem[Jain et~al.(2012)Jain, Kothari, and Thakurta]{jain2012differentially}
Jain, P., Kothari, P., and Thakurta, A.
\newblock Differentially private online learning.
\newblock In \emph{Conference on Learning Theory}, pp.\  24--1. JMLR Workshop
  and Conference Proceedings, 2012.

\bibitem[Jin et~al.(2021)Jin, McMurtry, Rubinstein, and Ohrimenko]{jin2021we}
Jin, J., McMurtry, E., Rubinstein, B.~I., and Ohrimenko, O.
\newblock Are we there yet? timing and floating-point attacks on differential
  privacy systems.
\newblock \emph{arXiv preprint arXiv:2112.05307}, 2021.

\bibitem[Joseph \& Bhatnagar(2015)Joseph and Bhatnagar]{joseph2015stochastic}
Joseph, A.~G. and Bhatnagar, S.
\newblock A stochastic approximation algorithm for quantile estimation.
\newblock In \emph{International Conference on Neural Information Processing},
  pp.\  311--319. Springer, 2015.

\bibitem[Joseph et~al.(2019)Joseph, Mao, Neel, and Roth]{8948625}
Joseph, M., Mao, J., Neel, S., and Roth, A.
\newblock The role of interactivity in local differential privacy.
\newblock In \emph{2019 IEEE 60th Annual Symposium on Foundations of Computer
  Science (FOCS)}, pp.\  94--105, 2019.
\newblock \doi{10.1109/FOCS.2019.00015}.

\bibitem[Kairouz et~al.(2015)Kairouz, Oh, and
  Viswanath]{kairouz2015composition}
Kairouz, P., Oh, S., and Viswanath, P.
\newblock The composition theorem for differential privacy.
\newblock In \emph{International conference on machine learning}, pp.\
  1376--1385. PMLR, 2015.

\bibitem[Kasiviswanathan et~al.(2011)Kasiviswanathan, Lee, Nissim,
  Raskhodnikova, and Smith]{kasiviswanathan2011can}
Kasiviswanathan, S.~P., Lee, H.~K., Nissim, K., Raskhodnikova, S., and Smith,
  A.
\newblock What can we learn privately?
\newblock \emph{SIAM Journal on Computing}, 40\penalty0 (3):\penalty0 793--826,
  2011.

\bibitem[Langley(2000)]{langley00}
Langley, P.
\newblock Crafting papers on machine learning.
\newblock In Langley, P. (ed.), \emph{Proceedings of the 17th International
  Conference on Machine Learning (ICML 2000)}, pp.\  1207--1216, Stanford, CA,
  2000. Morgan Kaufmann.

\bibitem[Lawson(2009)]{5370703}
Lawson, N.
\newblock Side-channel attacks on cryptographic software.
\newblock \emph{IEEE Security \& Privacy}, 7\penalty0 (6):\penalty0 65--68,
  2009.
\newblock \doi{10.1109/MSP.2009.165}.

\bibitem[Lee et~al.(2022)Lee, Liao, Seo, and Shin]{lee2022fast}
Lee, S., Liao, Y., Seo, M.~H., and Shin, Y.
\newblock Fast and robust online inference with stochastic gradient descent via
  random scaling.
\newblock In \emph{Proceedings of the AAAI Conference on Artificial
  Intelligence}, volume~36, pp.\  7381--7389, 2022.

\bibitem[Lei(2011)]{lei2011differentially}
Lei, J.
\newblock Differentially private m-estimators.
\newblock \emph{Advances in Neural Information Processing Systems}, 24, 2011.

\bibitem[Liu et~al.(2022)Liu, Sun, Kong, and Jiang]{IAMA}
Liu, Y., Sun, K., Kong, L., and Jiang, B.
\newblock Identification, amplification and measurement: A bridge to gaussian
  differential privacy.
\newblock \emph{Advances in Neural Information Processing Systems}, 2022.

\bibitem[Mironov(2012)]{mironov2012significance}
Mironov, I.
\newblock On significance of the least significant bits for differential
  privacy.
\newblock In \emph{Proceedings of the 2012 ACM conference on Computer and
  communications security}, pp.\  650--661, 2012.

\bibitem[Shao(2010)]{shao2010self}
Shao, X.
\newblock A self-normalized approach to confidence interval construction in
  time series.
\newblock \emph{Journal of the Royal Statistical Society: Series B (Statistical
  Methodology)}, 72\penalty0 (3):\penalty0 343--366, 2010.

\bibitem[Shao(2015)]{shao2015self}
Shao, X.
\newblock Self-normalization for time series: a review of recent developments.
\newblock \emph{Journal of the American Statistical Association}, 110\penalty0
  (512):\penalty0 1797--1817, 2015.

\bibitem[Smith(2011)]{smith2011privacy}
Smith, A.
\newblock Privacy-preserving statistical estimation with optimal convergence
  rates.
\newblock In \emph{Proceedings of the forty-third annual ACM symposium on
  Theory of computing}, pp.\  813--822, 2011.

\bibitem[Warner(1965)]{warner1965randomized}
Warner, S.~L.
\newblock Randomized response: A survey technique for eliminating evasive
  answer bias.
\newblock \emph{Journal of the American Statistical Association}, 60\penalty0
  (309):\penalty0 63--69, 1965.

\bibitem[Wei et~al.(2020)Wei, Li, Ding, Ma, Yang, Farokhi, Jin, Quek, and
  Vincent~Poor]{9069945}
Wei, K., Li, J., Ding, M., Ma, C., Yang, H.~H., Farokhi, F., Jin, S., Quek, T.
  Q.~S., and Vincent~Poor, H.
\newblock Federated learning with differential privacy: Algorithms and
  performance analysis.
\newblock \emph{IEEE Transactions on Information Forensics and Security},
  15:\penalty0 3454--3469, 2020.
\newblock \doi{10.1109/TIFS.2020.2988575}.

\end{thebibliography}
\bibliographystyle{icml2023}

\newpage
\appendix
\onecolumn


\section{Additional figures and tables}\label{Sec:c}

\begin{figure}[H]
    \centering
    \includegraphics[width= 8 cm,height=6cm]{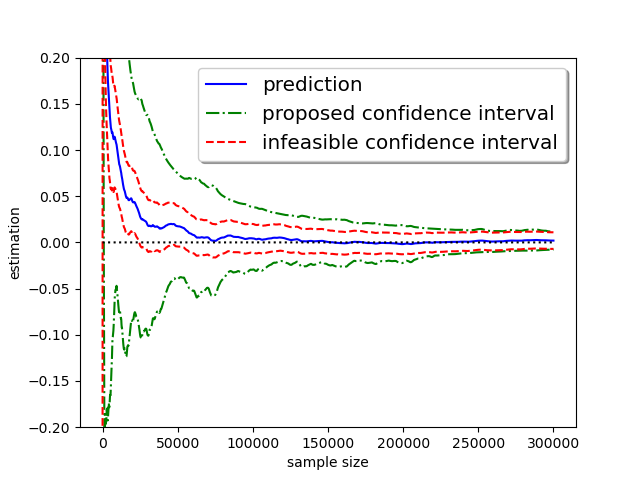}
    \caption{An alternative sample trajectory of estimator $Q_n$ using a different initialization $q_0=1$. }
    \label{fig:alterfigure1}
\end{figure}
\begin{figure}[H]
    \centering
    \includegraphics[width= 8 cm,height=6cm]{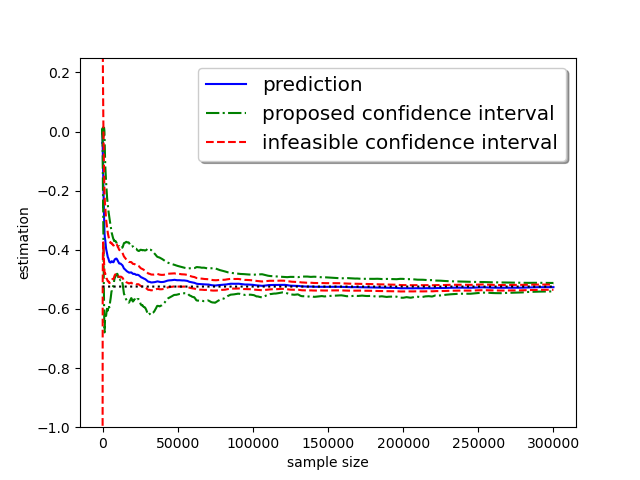}
    \caption{An alternative sample trajectory of estimator $Q_n$ using a different target quantile $\tau=0.3$. }
    \label{fig:alterfigure2}
\end{figure}

\begin{figure}[h!]
  \centering  \includegraphics[width= 8.4cm,height=8.4cm]{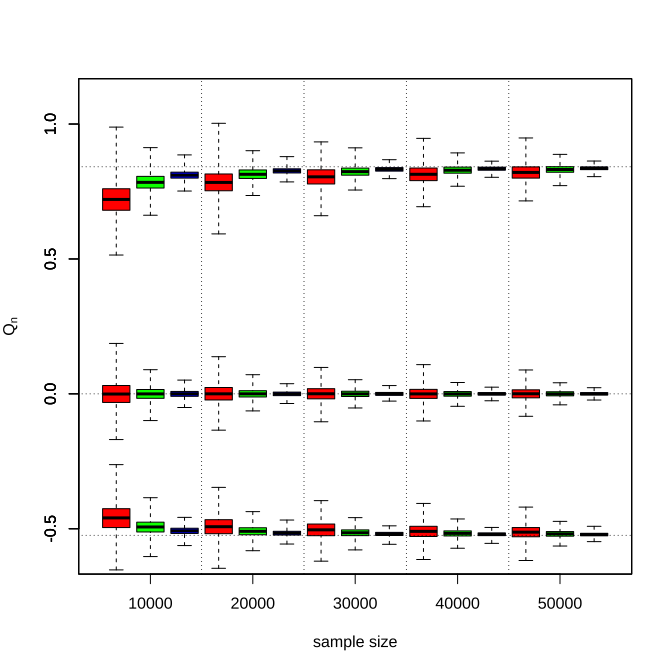}
  \caption{Box-plot of estimator $Q_n$ for different target quantile of Cauchy distribution. In each sample size divided by a vertical dotted line, the three boxes establish results with different privacy budgets by left: $r = 0.25$, middle $r = 0.5$, and right: $r = 0.9$. The horizontal dashed lines represent the true value $Q$ in $\tau = 0.3$, $0.5$, $0.8$ from the bottom to the top.}\label{Figure1}
\end{figure}
\begin{figure}[h!] 
\centering  \includegraphics[width= 8.4cm,height=8.4cm]
{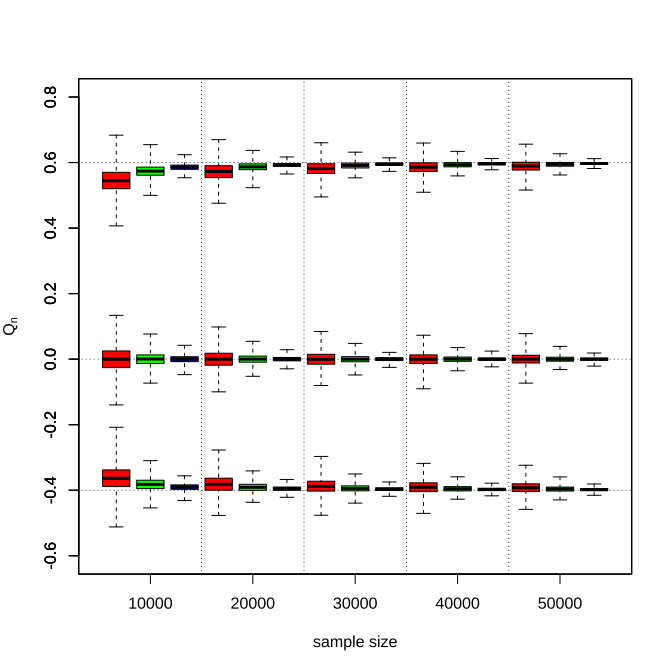}
  \caption{Box-plot of estimator $Q_n$ for different target quantile of Uniform distribution. In each sample size divided by a vertical dotted line, the three boxes establish results with different privacy budgets by left: $r = 0.25$, middle $r = 0.5$, and right: $r = 0.9$. The horizontal dashed lines represent the true value $Q$ in $\tau = 0.3$, $0.5$, $0.8$ from the bottom to the top.}\label{Figure3}
\end{figure}

\begin{figure}[h!]
  \centering  \includegraphics[width= 8.4cm,height=8.4cm]{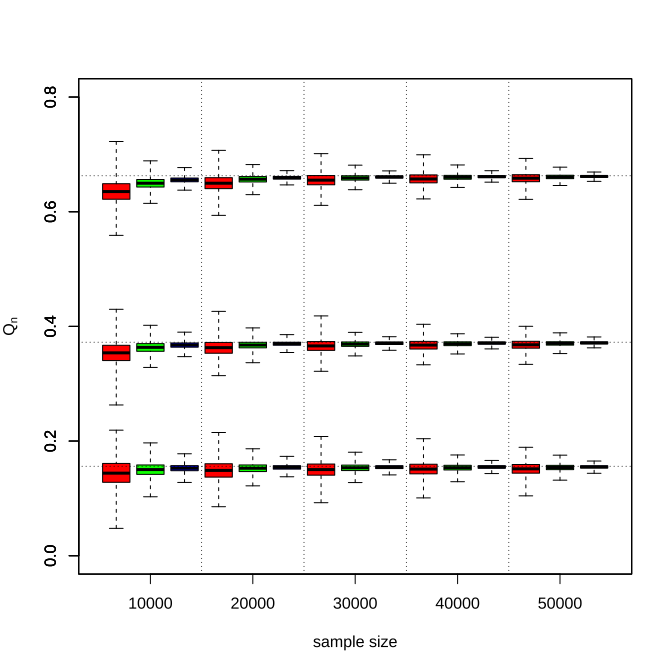}
  \caption{Box-plot of estimator $Q_n$ for different target quantile of PERT distribution. In each sample size divided by a vertical dotted line, the three boxes establish results with different privacy budgets by left: $r = 0.25$, middle $r = 0.5$, and right: $r = 0.9$. The horizontal dashed lines represent the true value $Q$ in $\tau = 0.3$, $0.5$, $0.8$ from the bottom to the top. }\label{Figure4}
\end{figure}

\begin{table}[h!]
  \centering
  \caption{Empirical results of coverage rate(mean absolute error) of proposed confidence interval (\ref{SNinerval}) (estimator $Q_n$) with data collected from Cauchy.}
\resizebox{85mm}{!}{
    \begin{tabular}{ccccc}\toprule
    $n$   & $\tau$ &  $r = 0.25$  &  $r=0.5$  &  $r=0.9$  \\\hline
    \multirow{3}[0]{*}{10000} & 0.3   & 0.894(0.140) & 0.972(0.068) & 0.987(0.037) \\
          & 0.5   & 0.807(0.045) & 0.876(0.024) & 0.906(0.014) \\
          & 0.8   & 0.853(0.399) & 0.989(0.207) & 1.000(0.112) \\\hline
    \multirow{3}[0]{*}{20000} & 0.3   & 0.928(0.076) & 0.966(0.037) & 0.982(0.020) \\
          & 0.5   & 0.872(0.034) & 0.908(0.018) & 0.927(0.010) \\
          & 0.8   & 0.950(0.219) & 0.991(0.105) & 0.998(0.055) \\\hline
    \multirow{3}[0]{*}{40000} & 0.3   & 0.944(0.044) & 0.964(0.022) & 0.974(0.012) \\
          & 0.5   & 0.900(0.025) & 0.926(0.012) & 0.939(0.007) \\
          & 0.8   & 0.965(0.114) & 0.984(0.053) & 0.993(0.028) \\\hline
    \multirow{3}[0]{*}{100000} & 0.3   & 0.944(0.025) & 0.956(0.012) & 0.963(0.007) \\
          & 0.5   & 0.927(0.016) & 0.935(0.008) & 0.945(0.004) \\
          & 0.8   & 0.956(0.054) & 0.970(0.026) & 0.980(0.013) \\\hline
    \multirow{3}[0]{*}{200000} & 0.3   & 0.948(0.017) & 0.954(0.008) & 0.958(0.004) \\
          & 0.5   & 0.936(0.011) & 0.944(0.006) & 0.945(0.003) \\
          & 0.8   & 0.952(0.034) & 0.966(0.017) & 0.971(0.008) \\\hline
    \multirow{3}[0]{*}{400000} & 0.3   & 0.942(0.012) & 0.954(0.006) & 0.952(0.003) \\
          & 0.5   & 0.944(0.008) & 0.949(0.004) & 0.946(0.002) \\
          & 0.8   & 0.948(0.023) & 0.960(0.011) & 0.961(0.005) \\\toprule
    \end{tabular}}
  \label{tab:cauchy}%
\end{table}%

\begin{table}[h!]
  \centering
  \caption{Empirical results of coverage rate(mean absolute error) of proposed confidence interval (\ref{SNinerval}) (estimator $Q_n$) with data collected from PERT.}
  \resizebox{85mm}{!}{
    \begin{tabular}{ccccc}\toprule
    $n$   & $\tau$ & $r = 0.25$ & $r=0.5$ & $r=0.9$ \\\hline
    \multirow{3}[0]{*}{10000} & 0.3   & 0.900(0.021) & 0.927(0.011) & 0.938(0.006) \\
          & 0.5   & 0.951(0.022) & 0.970(0.011) & 0.971(0.006) \\
          & 0.8   & 0.990(0.029) & 0.997(0.014) & 0.998(0.008) \\\hline
    \multirow{3}[0]{*}{20000} & 0.3   & 0.920(0.015) & 0.932(0.007) & 0.941(0.004) \\
          & 0.5   & 0.950(0.014) & 0.957(0.007) & 0.962(0.004) \\
          & 0.8   & 0.983(0.016) & 0.990(0.008) & 0.992(0.004) \\\hline
    \multirow{3}[0]{*}{40000} & 0.3   & 0.927(0.011) & 0.937(0.005) & 0.936(0.003) \\
          & 0.5   & 0.947(0.009) & 0.951(0.004) & 0.955(0.002) \\
          & 0.8   & 0.974(0.009) & 0.978(0.005) & 0.982(0.002) \\\hline
    \multirow{3}[0]{*}{100000} & 0.3   & 0.934(0.007) & 0.936(0.003) & 0.942(0.002) \\
          & 0.5   & 0.948(0.005) & 0.948(0.003) & 0.956(0.001) \\
          & 0.8   & 0.967(0.005) & 0.969(0.003) & 0.972(0.001) \\\hline
    \multirow{3}[0]{*}{200000} & 0.3   & 0.936(0.005) & 0.935(0.002) & 0.939(0.001) \\
          & 0.5   & 0.943(0.004) & 0.952(0.002) & 0.949(0.001) \\
          & 0.8   & 0.960(0.004) & 0.963(0.002) & 0.964(0.001) \\\hline
    \multirow{3}[0]{*}{400000} & 0.3   & 0.936(0.003) & 0.935(0.002) & 0.936(0.001) \\
          & 0.5   & 0.946(0.003) & 0.946(0.001) & 0.946(0.001) \\
          & 0.8   & 0.955(0.003) & 0.956(0.001) & 0.956(0.001) \\\toprule
    \end{tabular}}
  \label{tab:pert}%
\end{table}%

\begin{table}[h!]
  \centering
  \caption{Empirical results of coverage rate(mean absolute error) of proposed confidence interval (\ref{SNinerval}) (estimator $Q_n$) with data collected from Uniform.}
  \resizebox{85mm}{!}{
    \begin{tabular}{ccccc}\toprule
        $n$   & $\tau$ & $r = 0.25$ & $r=0.5$ & $r=0.9$ \\\hline
    \multirow{3}[0]{*}{10000} & 0.3   & 0.922(0.043) & 0.956(0.021) & 0.972(0.011) \\
          & 0.5   & 0.853(0.030) & 0.898(0.016) & 0.928(0.009) \\
          & 0.8   & 0.965(0.057) & 0.984(0.028) & 0.994(0.015) \\\hline
    \multirow{3}[0]{*}{20000} & 0.3   & 0.930(0.027) & 0.950(0.013) & 0.963(0.007) \\
          & 0.5   & 0.896(0.022) & 0.928(0.011) & 0.934(0.006) \\
          & 0.8   & 0.960(0.032) & 0.977(0.016) & 0.984(0.008) \\\hline
    \multirow{3}[0]{*}{40000} & 0.3   & 0.939(0.017) & 0.953(0.009) & 0.959(0.004) \\
          & 0.5   & 0.921(0.016) & 0.934(0.008) & 0.943(0.004) \\
          & 0.8   & 0.959(0.019) & 0.969(0.009) & 0.974(0.005) \\\hline
    \multirow{3}[0]{*}{100000} & 0.3   & 0.942(0.010) & 0.953(0.005) & 0.955(0.003) \\
          & 0.5   & 0.939(0.010) & 0.942(0.005) & 0.943(0.003) \\
          & 0.8   & 0.954(0.011) & 0.959(0.005) & 0.960(0.003) \\\hline
    \multirow{3}[0]{*}{200000} & 0.3   & 0.944(0.007) & 0.950(0.003) & 0.950(0.002) \\
          & 0.5   & 0.938(0.007) & 0.947(0.004) & 0.946(0.002) \\
          & 0.8   & 0.950(0.007) & 0.956(0.004) & 0.957(0.002) \\\hline
    \multirow{3}[0]{*}{400000} & 0.3   & 0.945(0.005) & 0.947(0.002) & 0.950(0.001) \\
          & 0.5   & 0.944(0.005) & 0.951(0.002) & 0.950(0.001) \\
          & 0.8   & 0.946(0.005) & 0.955(0.002) & 0.948(0.001) \\\toprule
    \end{tabular}}
  \label{tab:uniform}%
\end{table}%

\begin{table}[h!]
  \centering
  \caption{Conversion table between $r$ and $\epsilon$}
\begin{tabular}{cccc}
\hline
$r$  & $\epsilon$ & $r$  & $\epsilon$ \\ \hline
0    & 0          & 0.5  & 1.10       \\
0.05 & 0.10       & 0.55 & 1.24       \\
0.1  & 0.20       & 0.6  & 1.39       \\
0.15 & 0.30       & 0.65 & 1.55       \\
0.2  & 0.40       & 0.7  & 1.73       \\
0.25 & 0.51       & 0.75 & 1.95       \\
0.3  & 0.62       & 0.8  & 2.20       \\
0.35 & 0.73       & 0.85 & 2.51       \\
0.4  & 0.85       & 0.9  & 2.94       \\
0.45 & 0.97       & 0.95 & 3.66\\\hline      
\end{tabular}
  \label{reconversiontable}%
\end{table}%

\clearpage

\section{Alternative self-normalizes}\label{alternativeselfnorm}
 The following self-normalizer can also be used to construct the asymptotically pivotal quantity,
\begin{align*}
    N_n^{\prime} &=\sup_{t \in [0,1]}\left|S_{[nt]} - [nt]Q_n\right|,\\
    N_n^{\prime\prime} &= \int_{0}^{1}\left|S_{[nt]} - [nt]Q_n\right|dt,
\end{align*}
and based on the continuous mapping theorem again, one has that,
\begin{align*}
\frac{n^{-1/2}(S_{n}- nQ)}{n^{-1/2}N_n^{\prime}} &\xrightarrow{d}\frac{W(1)}{\sup_{t \in [0,1]}\left|W(t) - tW(1)\right|},\\
      \frac{n^{-1/2}(S_{n}- nQ)}{n^{-1/2}N_n^{\prime\prime}} &\xrightarrow{d} \frac{W(1)}{\int_0^1\left|W(t) - tW(1)\right|dt}.
\end{align*}
\section{Proof}
\subsection{Proof of Theorem \ref{THM:LRCDP}}
\label{proof:LRC}

Exhaustive computation yields that for any $(a,b)\in\{0,1\}^2$

\begin{equation}\label{EQ:KTDP}
    \frac{\mathbb{P}(LRC(q,r,x)=a|\mathbf{1}_{x>q}=b)}{\mathbb{P}(LRC(q,r,x)=a|\mathbf{1}_{x>q}=1-b)}\in\{ \frac{ 1+r} {1 - r},\frac{ 1-r} {1 + r}\}
\end{equation}

\subsection{Proof of Theorem \ref{THM:consist}, \ref{THM:normality} and \ref{THM:weakconvergence}}\label{proof:asym}
One can verify that the recursive equation (\ref{EQ:Qn}) is asymptotically equivalent to
\begin{align*}
    q_{n+1} &= q_n + d_n\left(1 - \frac{2}{1-r +2r(1-\tau)}{1}_{x^*_n>q_n}\right),
\end{align*}
where $\mathbb{P}(x^*_n = x_n) = r$, $\mathbb{P}(x^*_n = -\infty) = \mathbb{P}(x^*_n = \infty) = (1-r)/2$. Let
\begin{align*}
    H(z,X) &=1 - \frac{2}{1-r +2r(1-\tau)}{1}_{X>z}\\
    h(z,X) &= \mathbb{E}H(z,X) = 1 - \frac{2(1-F(z))}{1-r +2r(1-\tau)}.
\end{align*}
Hence $F(Q) = \tau$ is equivalent to $h(Q,X^*) = 0$. Then, one will find that the estimation of $Q$ with sample $x_1,\dots, x_n$ under LDP is equivalent to the estimation of $Q^*$ with sample $x_1,\dots, x_n$ without LDP constraints. The standard framework of the SGD method, such as Theorem 2 and 3 in \citep{dippon1998globally}, can be applied. Moreover, the statements in Theorems \ref{THM:consist}, \ref{THM:normality}, and \ref{THM:weakconvergence} hold true.

\subsection{}\label{APfinal}
We prove this by contradiction. Assuming that for any $n_0>1$ there is a $n>n_0$ such that :$$\epsilon_n/\sigma_n>k>1.$$
Let $$w=\Phi\left(-\frac{1}{2}\right)-\Phi\left(\frac{1}{2}-k\right)>0.$$ We choose a sufficiently large $n_0$ such that for any $n>n_0$ $$\mathbb{P}(\hat{\theta}_n/\sigma_n<1/2| H_0)\geq \Phi(1/2)-w/3$$ and  $$\mathbb{P}(\hat{\theta}_n/\sigma_n<1/2| H_1)\leq \Phi(1/2-\epsilon_n/\sigma_n)+w/3\leq \Phi(1/2-k)+w/3.$$

Then,
\begin{align*}
    \mathbb{P}(\hat{\theta}_n/\sigma_n<1/2| H_0)- \mathbb{P}(\hat{\theta}_n/\sigma_n<1/2| H_1)&\geq  \Phi(1/2)-\Phi(1/2-k)-2w/3\\
    &=2 \Phi\left(\frac{1}{2}\right)-1+\Phi\left(-\frac{1}{2}\right)-\Phi\left(\frac{1}{2}-k\right)-2w/3 \\
&=2 \Phi\left(\frac{1}{2}\right)-1 +w/3\\
&> 2 \Phi\left(\frac{1}{2}\right)-1+w/6
\end{align*}
Then $\hat{\theta}_n$ is not $(0,2 \Phi\left(\frac{1}{2}\right)-1+w/6)$-DP and therefore is not asymptotically $1$-GDP leading to a contradiction.
\end{document}